\documentclass[a4paper,USenglish,cleveref]{lipics-v2021}
\hideLIPIcs
\usepackage{graphicx} % Required for inserting images
\graphicspath{{figures/}}
\usepackage{xspace}
\usepackage{todonotes}

\usepackage{xcolor}
\definecolor{defblue}{rgb}{0.121,0.47,0.705}
\definecolor{linkblue}{rgb}{0.098,0.098,0.4392}
\let\emph\relax
\DeclareTextFontCommand{\emph}{\color{defblue}\em}
\DeclareTextFontCommand{\bl}{\color{defblue}}
\hypersetup{colorlinks=true,
    linkcolor=linkblue,
    anchorcolor=linkblue,
    citecolor=linkblue,
    filecolor=linkblue,
    menucolor=linkblue,
    urlcolor=linkblue,
    bookmarks=true,
    bookmarksopen=true,
    bookmarksopenlevel=2,
    bookmarksnumbered=true,
    hyperindex=true,
    plainpages=false,
    pdfpagelabels=true}
    
\newcommand{\true}{\text{true}}

\nolinenumbers

\title{The $st$-Planar Edge Completion Problem is Fixed-Parameter Tractable}

   \author{Liana Khazaliya}{Technische Universit\"at Wien, Austria}{lkhazaliya@ac.tuwien.ac.at}{https://orcid.org/0009-0002-3012-7240}{}

  \author{Philipp Kindermann}{FB IV - Informatikwissenschaften, Universität Trier, Germany}{kindermann@uni-trier.de}{https://orcid.org/0000-0001-5764-7719}{}
  
 \author{Giuseppe Liotta}{Department of Engineering, University of Perugia, Italy}{giuseppe.liotta@unipg.it}{https://orcid.org/0000-0002-2886-9694}{}

  \author{Fabrizio Montecchiani}{Department of Engineering, University of Perugia, Italy}{fabrizio.montecchiani@unipg.it}{https://orcid.org/0000-0002-0543-8912}{}
  
 \author{Kirill Simonov}{Hasso Plattner Institute, University of Potsdam, Germany}{kirill.simonov@hpi.de}{https://orcid.org/0000-0001-9436-7310}{}

\funding{Research of GL and FM partially supported  by MUR of Italy, under PRIN Project n. 2022ME9Z78 - NextGRAAL: Next-generation algorithms for constrained GRAph visuALization, and under PRIN Project n. 2022TS4Y3N - EXPAND: scalable algorithms for EXPloratory Analyses of heterogeneous and dynamic Networked Data.}

 \authorrunning{L. Khazaliya, P. Kindermann, G. Liotta, F. Montecchiani, K. Simonov} %TODO mandatory. First: Use abbreviated first/middle names. Second (only in severe cases): Use first author plus 'et al.'

 \Copyright{Liana Khazaliya, Philipp Kindermann, Giuseppe Liotta, Fabrizio Montecchiani, Kirill Simonov} %TODO mandatory, please use full first names. LIPIcs license is "CC-BY";  http://creativecommons.org/licenses/by/3.0/

 \acknowledgements{This research was initiated at Dagstuhl Seminar 23162: New Frontiers of Parameterized Complexity in Graph Drawing.}

%TODO mandatory: Please choose ACM 2012 classifications from https://dl.acm.org/ccs/ccs_flat.cfm 

 %TODO mandatory; please add comma-separated list of keywords

\newcommand{\skel}{\mathrm{skel}\xspace}

\newcommand{\stpga}{\textsc{$st$-Planar Edge Completion}\xspace}
\newcommand{\stshort}{\textsc{$st$-PEC}\xspace}
\newcommand{\defparquestion}[4]{
	\vspace{2mm}
	\noindent\fbox{
		\begin{minipage}{0.96\linewidth}
			\begin{tabular*}{\linewidth}{@{\extracolsep{\fill}}lr} #1 & \\ \end{tabular*}
			{\bf{Input:}} #2 \\
                {\bf{Parameter:}} #3 \\
			{\bf{Question:}} #4
		\end{minipage}
	}
	\vspace{2mm}
}

\begin{document}

\maketitle

\begin{abstract}
The problem of deciding whether a biconnected planar digraph $G=(V,E)$ can be augmented to become an $st$-planar graph by adding a set of oriented edges $E' \subseteq V \times V$  is known to be NP-complete. We show that the problem is fixed-parameter tractable when parameterized by the size of the set $E'$.    
\end{abstract}

\section{Introduction}
Edge modification problems have long been a subject of investigation in graph algorithms, resulting in a vast body of literature dedicated to exploring their computational complexity (refer, for instance, to Burzyn et al.~\cite{DBLP:journals/dam/BurzynBD06} and to Natanzon et al.~\cite{DBLP:journals/dam/NatanzonSS01} for comprehensive surveys). One specific category within this realm is the family of edge completion problems, which can be succinctly described as follows: Given a graph $G=(V,E)$ and a graph family $\mathcal{G}$, the objective is to determine whether it is possible to augment $G$ with a set $E' \subseteq V \times V$ of edges such that $G'=(V,E \cup E') \in \mathcal{G}$. In such cases, we say that $G$ becomes a member of $\mathcal{G}$ by \emph{adding} the edges in $E'$. Edge completion problems are frequently known to be NP-hard, thereby inspiring numerous studies focusing on parameterized complexity. For a comprehensive examination of parameterized algorithms addressing edge completion problems, we point the reader to the exhaustive survey by Crespelle et al.~\cite{DBLP:journals/csr/CrespelleDFG23}.

This paper focuses on the investigation of an edge completion problem specifically applied to directed graphs (\emph{digraphs} for short). More precisely, let $G=(V,E)$ be a digraph. A vertex of $G$ with no incoming edges is a \emph{source} of $G$, while a vertex without outgoing edges is a \emph{sink} of $G$. A digraph $G$ is an \emph{$st$-planar graph} if it admits a planar embedding such that: (1) it contains no directed cycle; (2) it contains a single source vertex $s$ and a single sink vertex~$t$; (3) $s$ and $t$ both belong to the external face of the planar embedding.

Upward planarity is a rather natural and well-studied notion of planarity for directed graphs (see, e.g.,~\cite{DBLP:conf/compgeom/ChaplickGFGRS22,DBLP:conf/gd/ChaplickGFGRS22,DBLP:books/ph/BattistaETT99,DBLP:journals/tcs/BattistaT88,DBLP:journals/siamcomp/GargT01,DBLP:journals/jct/TrotterM77}). In particular, a planar digraph is \emph{upward} if it admits a planar drawing where all edges are oriented upward.  A well-known result in graph drawing states that a digraph $G$ is upward if and only if  $G$ is a subgraph of an $st$-planar graph~\cite{DBLP:books/ph/BattistaETT99,DBLP:journals/tcs/BattistaT88}\footnote{From the proof in Lemma 4.1 of~\cite{DBLP:journals/tcs/BattistaT88}, one can in fact observe that a digraph is upward planar if and only if it is a subgraph of an $st$-planar graph defined over the same set of vertices.}. However, since testing for upward planarity is an NP-complete problem already for biconnected graphs~\cite{DBLP:journals/siamcomp/GargT01}, determining whether a biconnected graph is a subgraph of an $st$-planar graph is also computationally challenging. On the other hand, checking whether a digraph is $st$-planar can be done efficiently in polynomial time. This observation motivates for the investigation of the following problem.

\defparquestion{\stpga (\stshort)}{A biconnected digraph $G$}{$k \in \mathbb{N}$}{Is it possible to add at most $k$ edges to $G$ such that the resulting graph is an $st$-planar graph?}

\begin{figure}[t]
    \centering
    \subcaptionbox{}{\includegraphics[page=1]{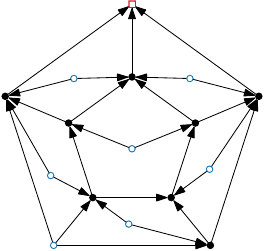}}
    \hfill
    \subcaptionbox{}{\includegraphics[page=2]{sources}}
    \hfill
    \subcaptionbox{}{\includegraphics[page=3]{sources}}
    \caption{(a) A digraph $G$ with $2k+1=7$ sources and 1 sink; $G$ has a unique planar embedding up to the choice of the external face; (b) A completion of $G$ to an $st$-planar graph obtained by adding $2k=6$ edges; (c) An upward planar drawing of the completion of $G$.}
    \label{fig:sources}
\end{figure}

In this paper, we present a fixed-parameter tractable algorithm for the \stpga problem. To help understanding the combinatorial and algorithmic challenges behind the problem, we make the observation that the parameter $k$ provides an upper bound on the number of sources and sinks in the input digraph $G$. Since an edge can remove the presence of at most one source and one sink, if the total number of sources and sinks in $G$ exceeds $2k + 2$, we can promptly reject the instance. Conversely, a positive answer to \stpga implies that $G$ is upward planar. In this respect, it is worth mentioning that Chaplick et al.~\cite{DBLP:conf/compgeom/ChaplickGFGRS22} have previously demonstrated that testing a digraph for upward planarity is fixed-parameter tractable when parameterized by the number of its sources. However, for every $k\ge 1$, there are upward planar digraphs with at most $2k+1$ sources that cannot be augmented to an $st$-planar graph by adding $k$ edges; refer to \cref{fig:sources} for an illustration. Furthermore, while an upward planarity test halts upon finding an upward planar embedding, not all upward planar embeddings of the same digraph can lead to an $st$-planar graph after the addition of $k$ edges. \cref{fig:free-embedding} demonstrates an upward planar digraph along with two of its upward planar embeddings: the embedding in \cref{fig:free-embedding-1} requires 6 edges to be augmented into an $st$-planar digraph, whereas the embedding in \cref{fig:free-embedding-3} can be augmented with 3 edges.

\begin{figure}[t]
    \centering
    \subcaptionbox{\label{fig:free-embedding-1}}{\includegraphics[page=1]{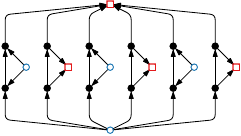}}
    \hfill
    \subcaptionbox{\label{fig:free-embedding-2}}{\includegraphics[page=3]{free-embedding}}
    \hfill
    \subcaptionbox{\label{fig:free-embedding-3}}{\includegraphics[page=6]{free-embedding}}
    \caption{(a) A biconnected digraph $G$ with $4$ sources and $4$ sinks; (b) With the given embedding, 6 edges have to be added to complete $G$ to an $st$-planar graph; (c) With a different embedding, adding 3 edges is sufficient.}
    \label{fig:free-embedding}
\end{figure}

In order to overcome the above technical challenges, our result is based on a structural decomposition of the digraph into its triconnected components using SPQR-trees (similarly as done in~\cite{DBLP:conf/compgeom/ChaplickGFGRS22}), as well as on novel insights regarding the combinatorial properties of upward planar digraphs. Since the proof is rather technical, after giving preliminaries and basic notation in \cref{se:preliminaries}, we present an overview of the approach in \cref{se:overview}. Next, the FPT algorithm is described in full detail in  \cref{se:main}. We conclude in \cref{se:open}.

\section{Preliminaries}\label{se:preliminaries}

In this section, we provide basic definitions and tools that will be used throughout the paper. 

\smallskip\noindent\textbf{Planar drawings and embeddings.} A \emph{planar drawing} of a graph $G$ maps the vertices of $G$ to points of the plane and the edges of $G$ to Jordan arcs such that no two arcs share a point except at common end-vertices. A planar drawing partitions the plane into topologically connected regions called \emph{faces}, one of which is unbounded and called the \emph{external face}, in contrast with all other faces which are \emph{inner faces}. For a digraph $G$, a planar drawing is called \emph{upward} if each edge oriented from a vertex $u$ to a vertex $v$ is represented by a Jordan arc monotonically increasing from the point representing $u$ to the point representing $v$.  A graph (digraph) is \emph{planar} (\emph{upward planar}) if it admits a planar drawing (upward planar drawing). A  \emph{planar embedding} (\emph{upward planar embedding}) $\mathcal{E}(G)$ of a planar graph (upward planar digraph) $G$ represents an equivalence class of planar drawings (upward planar drawings) with the same inner faces and the same external face, up to a homeomorphism of the plane. Graph $G$ is \emph{plane} if it comes with a fixed planar embedding $\mathcal{E}(G)$.

\smallskip\noindent\textbf{SPQR-trees.} We recall the definition of  \emph{SPQR-tree}, introduced in~\cite{DBLP:books/ph/BattistaETT99}, which represents the decomposition of a biconnected graph $G$ into its triconnected components~\cite{DBLP:journals/siamcomp/HopcroftT73}.
Each triconnected component corresponds to a non-leaf node $\nu$ of~$T$; the triconnected component itself is called the \emph{skeleton} of $\nu$ and is denoted as $\skel(\nu)$. Node $\nu$ can be: $(i)$ an \emph{R-node}, if $\skel(\nu)$ is a triconnected graph; $(ii)$ an \emph{S-node}, if $\skel(\nu)$ is a simple cycle of length at least three; $(iii)$ a \emph{P-node}, if $\skel(\nu)$ is a bundle of at least three parallel edges.
A degree-1 node of $T$ is a \emph{Q-node} and represents a single edge of~$G$.
A \emph{real edge} (resp. \emph{virtual edge})  in $\skel(\nu)$ corresponds to a Q-node (resp., to an S-, P-, or R-node) adjacent to $\nu$ in $T$.
Neither two S- nor two P-nodes are adjacent in~$T$. The SPQR-tree of a biconnected graph can be computed in linear time~\cite{DBLP:books/ph/BattistaETT99,DBLP:conf/gd/GutwengerM00}. Let $e$ be a designated edge of $G$, called the \emph{reference edge} of $G$, let $\rho$ be the Q-node of $T$ corresponding to $e$, and let $T$ be rooted at $\rho$. For any P-, S-, or R-node $\nu$ of $T$, $\skel(\nu)$ has a virtual edge, called \emph{reference edge} of $\nu$ and denoted as $e_\nu$, associated with a virtual edge in the skeleton of its parent. The end-vertices of the reference edge of $\nu$ are called the \emph{poles} of $\nu$. For every node $\nu \neq \rho$, the \emph{pertinent graph $G_\nu$ of $\nu$} is the subgraph of $G$ whose edges correspond to the Q-nodes in the subtree of $T$ rooted at $\nu$.  Without loss of generality, we shall consider SQPR-trees where every S-node has exactly two children (see, e.g., \cite{DBLP:conf/compgeom/ChaplickGFGRS22,DBLP:journals/siamcomp/BattistaLV98,DBLP:conf/gd/DidimoKLO22}); this lifts the condition that two S-nodes cannot be adjacent in $T$.

\smallskip\noindent\textbf{Angles in upward drawings.} Let $G=(V,E)$ be a digraph. For each edge $(u,v) \in E$, we write $uv$ if $(u,v)$ is oriented from $u$ to $v$ in $G$, and we write $vu$ otherwise. A vertex $v$ is a \emph{switch} of $G$, if it is either a source or a sink, and it is a \emph{non-switch} otherwise. Recall that a digraph is upward planar if and only if it is a subgraph of an $st$-planar graph~\cite{DBLP:books/ph/BattistaETT99}. Hence, being upward planar is a necessary condition for YES-instances of \stpga. Consider now a biconnected plane digraph $G$. An \emph{angle} is an incidence between a vertex $v$ and a face $f$ of $G$. Let $\alpha$ be one such angle, and consider the two edges incident to $v$ that belong to the boundary of $f$. If such edges are one incoming and one outgoing, $\alpha$ is a \emph{non-switch angle}, while if the edges are both incoming or both outgoing, $\alpha$ is a \emph{switch angle}. Note that a switch angle in a face $f$ can be made by two edges that are incident to a non-switch vertex $v$: it is enough that the edges of $f$ incident to $v$ are both incoming or both outgoing. In this case, $v$ is a \emph{local switch} for face $f$. An \emph{angle assignment} is a labeling $\lambda$ of the angles of $G$ with labels $\{-1,0,+1\}$ (see, e.g.,~\cite{DBLP:journals/algorithmica/BertolazziBD02,DBLP:journals/algorithmica/BertolazziBLM94,DBLP:conf/gd/BinucciGLT21,DBLP:journals/siamdm/DidimoGL09}). In particular,  non-switch angles can only receive the label $0$, while switch angles can be labeled as either $-1$ or $+1$. The planar embedding of $G$ can be realized as an upward drawing if and only if there is an angle assignment such that: (i) each switch vertex has exactly one angle labeled $+1$; (ii) each non-switch vertex has exactly two angles labeled as $0$, while all the others are switch angles labeled  $-1$; (iii) the difference between the number of angles labeled  $+1$ and the number of angles labeled $-1$ along the boundary of each inner face  is $-2$; (iv) the difference between the number of angles labeled  $+1$ and the number of angles labeled  $-1$ along the boundary  of the external face is $+2$. Observe that property (ii) implies that each non-switch vertex forms exactly two non-switch angles. An angle assignment satisfying the above properties is called \emph{upward}. The restriction of an upward angle assignment to the angles of a single face $f$ is an \emph{upward angle assignment for $f$}.

\section{Overview of the Approach}\label{se:overview}

Let $G$ be a biconnected digraph. Since testing for planarity can be done in linear time, we shall assume that $G$ is planar. We begin by explaining two key ingredients for our algorithm, namely, the use of SPQR-trees to encode all the planar embeddings of $G$, and the use of upward angle assignments to incrementally saturate $G$. The main crux of our algorithm lies in blending these two ingredients together to design a dynamic program that solves the problem in FPT time. 

Let $T$ be a rooted SPQR-tree of a planar graph $G$ with reference edge $e$. The planar embeddings of $G$ in which the edge $e$ lies on the boundary of the external face can be obtained as follows (see, e.g.,~\cite{DBLP:books/ph/BattistaETT99}). For a P- or R-node $\nu$, denote by $\skel^-(\nu)$ the skeleton of $\nu$ without its reference edge.
If $\nu$ is a P-node, the embeddings of $\skel(\nu)$ are the different permutations of the edges of $\skel^-(\nu)$. If $\nu$ is an R-node, $\skel(\nu)$ has two possible embeddings, obtained by flipping $\skel^-(\nu)$ at its poles. No operations are needed at S- and Q-nodes.

Consider now an upward planar drawing $\Gamma$ of $G$ and hence assume that $G$ is plane. Let $\lambda$ be the upward angle assignment  \emph{induced} by $\Gamma$. Precisely, the switch angles that are larger (smaller) than $\pi$ in $\Gamma$ are labeled as $+1$ ($-1$), while the non-switch angles are labeled as $0$. Let $v$ be a source (sink) of $G$ and let $f$ be the face of $G$ in which $v$ makes its $+1$ angle. Let $u$ be a vertex of $f$ different from $v$.  We say that adding $uv$ ($vu$) to $G$ \emph{saturates} $v$, and that $uv$ ($vu$) is a \emph{saturating edge}. Namely, $v$ becomes a non-switch vertex in $G'=(V, E \cup \{uv\})$. Notably, $f$ is the only face in which an edge saturating $v$ can be added: one easily verifies that choosing any other face would lead to a non-upward angle assignment. 

Based on the previous reasoning, at high-level, the algorithm will exploit a bottom-up traversal of the SPQR-tree $T$ to explore the planar embeddings of $G$. For each visited node, it will keep track of the information related to the minimum number of edges required to saturate all switches that lie in the inner faces of the corresponding pertinent graph. The interface of a candidate solution is encoded in terms of ``signatures'' which, informally, are strings containing all  switches along the boundary of the external face of the pertinent graph that do not yet have any angle labeled as $+1$ and all vertices that must instead contribute with a $-1$ angle along the boundary. A running time bounded by a function of the budget $k$ is obtained by several crucial insights about how a bounded number of switches in the graph affects the possible signatures and limits the relevant embeddings to be considered. %Some of these insights are likely to be useful to solve other problems on upward planar graphs. 

\section{An FPT Algorithm for \stpga}\label{se:main}

In this section, we describe our FPT algorithm, which leads to the following theorem.

\begin{theorem}\label{th:main-bic}
Let $G$ be an $n$-vertex biconnected plane digraph.  There is an algorithm that solves \stpga in $2^{O(k^2)} \cdot n^2$ time.
\end{theorem}

We begin by describing the records used by our dynamic program (\cref{sse:records}), which are used to encode the angles along the boundary of the external face of a pertinent graph. Next, we describe the algorithm (\cref{sse:algo}), which constructs such records while traversing bottom-up the SPQR-tree of the input graph.

\begin{figure}
    \centering
    \includegraphics{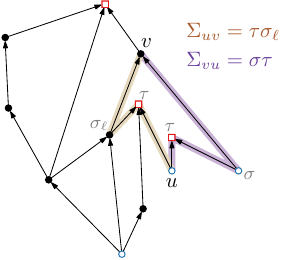}
    \caption{The signatures of two paths $\Pi_{uv}$ (brown background) and $\Pi_{vu}$ (purple background).}
    \label{fig:half}
\end{figure}

\subsection{Setting up the Records for Dynamic Programming}\label{sse:records}

\noindent\textbf{Signatures.} We begin with some notation and definitions. Let $G$ be a plane digraph. Let $\Pi_{uv}$ be a simple undirected path of $G$ from a vertex $u$ to a vertex $v$. The \emph{signature} of $\Pi_{uv}$ is a string $\Sigma_{uv}$ computed as follows. Consider a walk along $\Pi_{uv}$ from $u$ to $v$. For each encountered vertex $w$ distinct from $u$ and $v$, look at the two edges incident to $w$ in $\Pi_{uv}$. If the two edges are one incoming and one outgoing, we do not append any symbol to $\Sigma_{uv}$. If the two edges are both outgoing (incoming) and $w$ is a switch of $G$, we append the symbol~$\sigma$~($\tau$). If the two edges are both outgoing (incoming) and $w$ is not a switch of $G$  -- hence, it is a local switch for some face $f$ --, we append the symbol $\sigma_\ell$ ($\tau_\ell$). Observe that, if $\Pi_{uv}$ is a single edge connecting $u$ to $v$, then $\Sigma_{uv}=\emptyset$. At high level, the idea is that when walking along a piece of the boundary of some face $f$ of $G$,  non-switch angles are ignored as their only possible value in an angle assignment is $0$. On the other hand, the symbols $\sigma_\ell$ and $\tau_\ell$ will encode switch-angles whose only possible value is $-1$ (else the corresponding vertex would be a switch of $G$). Finally, the symbols $\sigma$ and $\tau$ will point to switch angles that may be assigned either $-1$ or $+1$. Refer to \cref{fig:half} for an illustration. %Observe that, since there can be neither two consecutive symbols in $\{\sigma,\sigma_\ell\}$  nor two consecutive symbols in $\{\tau,\tau_\ell\}$, each string contains a sequence of alternating (distinct) symbols. 

A signature is \emph{short} if it contains at most $4k+2$ symbols.  Let $\Sigma^*$ be the set of short signatures; we observe the following.

  \begin{observation}\label{ob:card}
      The cardinality of $\Sigma^*$ is $2^{O(k)}$.
  \end{observation}
%  \begin{proof}
%      A short sequence contains at most $4k+2$ symbols in $\{\sigma,\sigma_\ell,\tau,\tau_\ell\}$. Hence the possible short signatures are strings of length $O(k)$ defined over a set of $4$ symbols. 
 % \end{proof}

\noindent\textbf{Half-boundaries.}  Let $G$ be a biconnected planar digraph, and let $T$ be the SPQR-tree of $G$ rooted at a Q-node representing an arbitrary edge $e$ of $G$. For each node $\nu$ of $T$, we recall that $G_\nu$ is the pertinent graph, and we denote by $u,v$ the poles of $\nu$ (omitting the dependency on $\nu$ for simplicity). Assuming that $G_\nu$ comes with a fixed planar embedding, let $f$ be the external face of $G_\nu$. The \emph{half-boundary} $B_{uv}$ of $\nu$ is the path containing the vertices of $f$ encountered in a clockwise walk of the face from $u$ to $v$. The half-boundary $B_{vu}$ of $\nu$ is defined analogously walking from $v$ to $u$. A vertex $w$ on the boundary of $f$ is \emph{bifacial} if it belongs to both $B_{uv}$ and $B_{vu}$ (which implies that $w$ is a cutvertex of $G_\nu$ and hence $\nu$ is an S-node). For each of the two half-boundaries we can define the two corresponding signatures $\Sigma_{uv}$ and $\Sigma_{vu}$. We will assume that for each symbol of $\Sigma_{uv}$ and $\Sigma_{vu}$ we have a pointer to the corresponding vertex. Let $B$ be one of the two half-boundaries of $\nu$ and let $\Sigma$ be its signature. Let $B'$ be a path contained in $B$ (possibly $B=B'$). The \emph{restriction} of $\Sigma$ to $B'$, denoted as $\Sigma[B']$, is the substring of $\Sigma$ containing the symbols whose corresponding vertices belong to $B'$. The next lemma shows that working with short signatures is not restrictive. 

\begin{lemma}\label{le:short}
    Let $G$ be a biconnected upward planar digraph with a fixed upward planar embedding $\mathcal{E}(G)$. Let $T$ be the SPQR-tree of $G$ rooted at a Q-node representing an arbitrary edge $e$ of $G$. Let $\nu$ be a node of $T$. For any fixed $k$, if $G$ can be augmented to an $st$-planar graph by adding at most $k$ saturating edges, then the signatures $\Sigma_{uv}$ and $\Sigma_{vu}$ of the two half-boundaries $B_{uv}$ and $B_{vu}$ of $\nu$ are both short. 
\end{lemma}
\begin{proof}
    Let $\Gamma$ be an upward planar drawing of $G$ whose corresponding upward planar embedding is $\mathcal{E}(G)$, and consider the subdrawing $\Gamma'$ induced by $G_\nu$. Let $\lambda$ be the upward angle assignment induced by $\Gamma'$, and let $f$ be the external face of $G_\nu$. We know that $f$ contains at most $2k+2$ switches, otherwise $k$ saturating edges would not suffice to turn $G$ into an $st$-planar graph. Hence, $\lambda$ can label $+1$ at most $2k+2$  angles along the boundary of $f$. Also, since $\lambda$ obeys to property (iv) of an upward angle assignment, it labels $-1$ at most $2k$ angles. Therefore,  $\Sigma_{uv}$ and $\Sigma_{vu}$ can each contain at most $4k+2$ symbols. 
\end{proof}

\smallskip\noindent\textbf{Internal assignments.}  An angle of $G_\nu$ is \emph{internal} if it is defined in an inner face of $G_\nu$. An \emph{internal assignment} of $G_\nu$ is an  angle assignment $\lambda$ that labels all the internal angles of $G_\nu$ and that respects properties (i)--(iii) for upward angle assignments (but ignoring property (iv)). A switch vertex of $G$ is called \emph{active} with respect to $\lambda$ if none if its internal angles (if any) received value $+1$. The \emph{cost} of an internal assignment $\lambda$ of $G_\nu$ is the minimum number of saturating edges needed to saturate all switches of $G_\nu$ that are not active with respect~to~$\lambda$.  

\smallskip\noindent\textbf{Partial solutions.}   We are now ready to define the table used by our dynamic program. A tuple $\langle \Sigma_1, \Sigma_2, b_1, b_2 \rangle$, such that $\Sigma_1, \Sigma_2$ is a pair of short signatures and $b_1,b_2$ is a pair of flags, is called a \emph{candidate tuple} in the following. Given a node $\nu$ and a candidate tuple $\langle \Sigma_1, \Sigma_2, b_1, b_2 \rangle$, the function $X(\nu,\Sigma_1,\Sigma_2,b_u,b_v)$ returns the minimum cost of an internal assignment $\lambda$ of $G_\nu$ such that: (1) $\Sigma_1$ and $\Sigma_2$ are the signatures of its two half-boundaries $B_{uv}$ and $B_{vu}$, respectively, (2) the flag $b_u$  is true if and only if $u$ is an active switch with respect to $\lambda$, (3) the flag $b_v$ is true if and only if $v$ is an active switch with respect to $\lambda$. The set of \emph{partial solutions} for $\nu$ is given by the restriction of $X$ to the single node $\nu$. Also, a pair of signatures is \emph{empty} if both its signatures are empty (i.e., they do not contain any symbol).

\subsection{Description of the Algorithm}\label{sse:algo}

The function $X$ is computed by traversing $T$ bottom-up. For each node $\nu$ of $T$, we initialize $X(\nu, \Sigma_1, \Sigma_2, b_1, b_2)=+\infty$ for each candidate tuple $\langle \Sigma_1, \Sigma_2, b_1, b_2 \rangle$. We only ensure that $X(\nu, \Sigma_1, \Sigma_2, b_1, b_2)$ is computed precisely if the value is at most $k$; for any value larger than $k$ we assume that $X(\nu, \Sigma_1, \Sigma_2, b_1, b_2)=+\infty$ is the correct setting, since we are only interested in the solutions that add at most $k$ edges.

If $\nu$ is a leaf node, then it is a Q-node and $G_\nu$ is a single edge. In this case, either $u$ is the  source and $v$ is the sink of $G_\nu$, or vice-versa. Then we set $X(\nu,\emptyset,\emptyset,\true,\true)=0$.

The lemma below deals with the case in which $\nu$ is an S-node. Since S-nodes have exactly two children and are not used to describe the planar embeddings of $G$, the routine of the algorithm at S-nodes is relatively simple. Next, we will consider P-nodes and R-nodes, which require more involved arguments. 

\begin{lemma}\label{le:series}
    Let $\nu$ be an S-node of $T$. The set of partial solutions of $\nu$ can computed in $2^{O(k)}$ time.
\end{lemma}
\begin{proof}
    Let $\mu_1$ and $\mu_2$ be the two children of $\nu$.  In order to compute the partial solutions for $\nu$, we check whether pairs of internal assignments of $G_{\mu_1}$ and $G_{\mu_2}$ can be combined together. Let $\langle \Sigma_{1,1}, \Sigma_{1,2}, b_{1,1}, b_{1,2} \rangle$ and $\langle \Sigma_{2,1}, \Sigma_{2,2}, b_{2,1}, b_{2,2} \rangle$ be a pair of candidate tuples. Also, let $C=X(\mu_1, \Sigma_{1,1}, \Sigma_{1,2}, b_{1,1}, b_{1,2})+X(\mu_2, \Sigma_{2,1}, \Sigma_{2,2}, b_{2,1}, b_{2,2})$.   

    We first verify that $C \le k$, and that $b_{1,2} \vee b_{2,1} = \true$. The first condition guarantees that we have not exceeded our budget $k$ of saturating edges, while the second condition guarantees that the pole shared by $\mu_1$ and $\mu_2$ does not receive the value $+1$ twice in the final internal assignment of $G_\nu$. If both conditions are satisfied, then we proceed as detailed below, otherwise, we discard the pair of candidate tuples. 

    Denote by $u$ and $w$  the poles of $\mu_1$, and by $w$ and $v$ the poles of $\mu_2$. Observe that $B_{uv}$ corresponds to the union of $B_{uw}$ and $B_{wv}$ (vertex $w$ is hence bifacial). Based on this observation, we show how to compute $\Sigma_1$ for $B_{uv}$, the computation of $\Sigma_2$ can be performed analogously. We initially set $\Sigma_1 = \Sigma_{1,1}$. Consider the two edges incident to $w$ along $B_{uv}$. If one edge is incoming and the other is outgoing, then we do not append any symbol. If both edges are incoming or both outgoing, we check whether  one of $b_{1,2}$ and $b_{2,1}$ is false. If so, we append the symbol $\sigma_\ell$ if $w$ is a source of $G$, or the symbol $\tau_\ell$ otherwise. If none of $b_{1,2}$ and $b_{2,1}$ is false, we append the symbol $\sigma$ if $w$ is a source of $G$, or the symbol $\tau$ otherwise. Next, we concatenate the signature $\Sigma_{2,1}$. Once both $\Sigma_1$ and $\Sigma_2$ have been computed, we verify that each of them is short (a necessary condition by \cref{le:short}), otherwise we discard the candidate tuples. Finally,  we set $X(\nu, \Sigma_1, \Sigma_2, b_{1,1}, b_{2,2}) = C$. 
    
    By \cref{ob:card}, we have $2^{O(k)}$ possible pairs of signatures to consider, and performing the above operations takes $O(k)$ time for each pair.  
\end{proof}

The next tools will be useful for the remaining lemmas. The next result is based on the fact that face boundaries containing long sequences of non-switch vertices are irrelevant for the sake of computing the least number of saturating edges; see \cref{fig:saturating} for an illustration.

\begin{lemma}
    \label{le:saturating}
    Let $f$ be an inner face of $G$ with $n_f$ vertices, and let $\lambda_f$ be an upward angle assignment for $f$ with $h$ switch-angles. The minimum number of edges that saturate all switch vertices of $G$ forming an angle labeled $+1$ in $f$ can be computed in $O(2^{O(h^2)} + n_f)$~time. 
\end{lemma}
\begin{proof}
    Consider a sequence of non-switch angles in a walk along the boundary of $f$, and let $F$ be the corresponding set of vertices. Also, consider a set of saturating edges $S$ drawn inside $f$, such that each of them has one end-vertex that belongs to $F$; see \cref{fig:app-saturating-1}. Since in any upward drawing of $G$ the set $F$ is drawn as a monotonically increasing curve, it is immediate to see that the set $S$ can be replaced with a set of saturating edges $S'$ such that: $S$ and $S'$ have the same size, $S$ and $S'$ saturate the same set of switches, all edges of $S'$ are incident to a single (arbitrarily chosen) vertex of $F$; see \cref{fig:app-saturating-2}. Consequently, we can work with a simplified boundary of $f$ in which maximal sequences of vertices forming non-switch angles are replaced with a single vertex; see \cref{fig:app-saturating-3}. Computing the simplified boundary takes $O(n_f)$ time, and such a boundary has at most $2h+1$ vertices. Thus, the maximum number of edges whose end-vertices belong to this boundary is $\binom{2h+1}{2}$, and in $2^{O(h^2)}$ time we can enumerate all candidate sets of saturating edges whose end-vertices belong to this boundary. Finally, among these sets, we return the size of the smallest set that saturates all switches labeled  $+1$ and such that no two of its edges cross. (Whether two edges cross only depends on the order of their end-vertices along the boundary of $f$.)
\end{proof}

\begin{figure}
    \centering
    \subcaptionbox{\label{fig:saturating-1}}{\includegraphics[page=1]{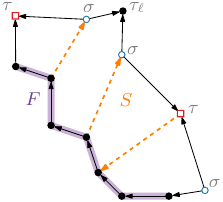}}
    \hfill
    \subcaptionbox{\label{fig:saturating-2}}{\includegraphics[page=2]{saturating}}
    \hfill
    \subcaptionbox{\label{fig:saturating-3}}{\includegraphics[page=3]{saturating}}
    \caption{Illustration for \cref{le:saturating}.}
    \label{fig:saturating}
\end{figure}

\begin{lemma}\label{le:signature}
    Let $\nu$ be a node of $T$ and let $\mu$ be a child of $\nu$. Suppose that $G_\nu$ is plane and a half-boundary $B$ of $\nu$ contains a half-boundary $B'$ of $\mu$ ($B$ and $B'$ may possibly coincide). Given an internal assignment $\lambda$ of $G_\nu$ and the signature of $B'$, the restriction of the signature of $B$ to  $B'$ can be computed in $O(k)$ time.
\end{lemma}
\begin{proof}
    Let $\Sigma'$ be the signature of $B'$, we compute the desired signature $\Sigma$ as follows.  If $\Sigma'$ does not contain any symbol in $\{\sigma,\tau\}$ whose corresponding vertex is bifacial, then $\Sigma=\Sigma'$. Otherwise we initialize $\Sigma=\Sigma'$ and proceed as follows. For each $\sigma$ or $\tau$ whose corresponding vertex $w$ is bifacial and incident to an inner face $f$ of $G_\nu$, we verify whether $\lambda$ has labeled $+1$ the angle that $w$ makes in $f$. If so, we replace the symbol $\sigma$ or $\tau$ with $\sigma_\ell$ or $\tau_\ell$, respectively. By construction, $\Sigma$ is the restriction of the signature of $B$ to $B'$.
\end{proof}

The next result will be used to bound the number of interesting children of a P-node.

\begin{figure}
    \centering
    \subcaptionbox{\label{fig:app-saturating-1}}{\includegraphics[page=1]{saturating}}
    \hfill
    \subcaptionbox{\label{fig:app-saturating-2}}{\includegraphics[page=2]{saturating}}
   \hfill
   \subcaptionbox{\label{fig:app-saturating-3}}{\includegraphics[page=3]{saturating}}
  \caption{Illustration for \cref{le:saturating}.}
   \label{fig:app-saturating}
\end{figure}

\begin{lemma}
    \label{le:empty}
    Let $\nu$ be a P-node of $T$ with poles $u,v$. Suppose that $G_\nu$ is plane, and let $\mu$ and $\mu'$ be two children of $\nu$ none of which is a Q-node, and whose corresponding edges of $\skel(\nu)^-$ are consecutive in the permutation fixed by the planar embedding of $G_\nu$. Also, suppose that for both $\mu$ and $\mu'$ it holds that the pair of signatures of its two half-boundaries is empty. Let $G'$ be the digraph obtained from $G$ by removing all vertices of $G_{\mu'}$ except the poles $u,v$. Then $G$ is a YES-instance of \stshort if and only if $G'$ is a YES-instance. 
\end{lemma}
\begin{proof}
    Let $F$ be the set of vertices of $G_{\mu'}$ distinct from $u$ to $v$ (i.e., those removed when going from $G$ to $G'$). Also, suppose that the half-boundaries forming an inner face of $G_\nu$ are $B_{vu}$ of $\mu$ and $B_{uv}$ of $\mu'$.     
    Suppose first that $G$ admits a solution, namely we can add a set $E'$ of $k$ saturating edges to $G$ and turn it into an $st$-planar graph. If none of these edges is such that exactly one end-vertex belongs to $F$, then $E'$ is a solution also for $G'$. Otherwise, let $S \in E'$ be the set of edges having exactly one end-vertex in $F$. Observe that the end-vertices in $F$ all belong to the half-boundary $B_{vu}$ of $\mu'$. Hence, since neither $\mu$ nor $\mu'$ is a Q-node, analogously as in the proof of \cref{le:saturating}, we can replace $S$ with a set $S'$ of edges incident to an arbitrary vertex of the half-boundary $B_{vu}$ of $\mu$.  This yields a new set $E''$ of saturating edges that corresponds to a solution for $G'$.

    Suppose now that $G'$ admits a solution $E'$. If $E'$ does not contain any edge having exactly one end-vertex in $B_{vu}$ of $\mu$, then $E'$ is a solution for $G$. Otherwise, let $S \in E'$ be the set of edges having exactly one end-vertex in $B_{vu}$ of $\mu$. again, since neither $\mu$ nor $\mu'$ is a Q-node, analogously as in the proof of \cref{le:saturating}, we can replace $S$ with a set $S'$ of edges incident to an arbitrary vertex of the half-boundary $B_{vu}$ of $\mu'$.  This yields a new set $E''$ of saturating edges that corresponds to a solution for $G$.
\end{proof}

We are now ready to deal with P- and R-nodes. 

\begin{lemma}\label{le:parallel}
   Let $\nu$ be a P-node of $T$. The set of partial solutions of $\nu$ can be computed in $2^{O(k^2)} \cdot n$ time.
\end{lemma}
\begin{proof}
    Let $u$ and $v$ be the poles of $\nu$, and let $\mu_1,\mu_2, \dots, \mu_h$ be the $h \ge 2$ children of $\nu$. In order to compute the partial solutions for $\nu$, similarly as for S-nodes, we check whether sets of internal assignments of $G_{\mu_1}$, $G_{\mu_2}$, $\dots$, $G_{\mu_h}$ can be combined together. For each child $\mu_i$, let $\langle \Sigma_{1,i}, \Sigma_{2,i}, b_{1,i}, b_{2,i} \rangle$ be a candidate tuple.  Let $C = \sum_{i=1}^h X(\mu_i, \Sigma_{1,i}, \Sigma_{2,i}, b_{1,i}, b_{2,i})$.

    We first verify that $C \le k$, and that at most one flag $b_{1,i}$ is true, as well as at most one flag $b_{2,i}$ is true. The first condition guarantees that we have not exceeded our budget $k$, while the second condition guarantees that the poles $u,v$ shared by the children of $\nu$ do not receive the value $+1$ twice in the final internal assignment. If both conditions are satisfied, then we proceed as detailed below, otherwise we discard the set of candidate tuples. 

    Observe that $h$ might be unbounded with respect to $k$, thus we cannot afford to enumerate all possible permutations of the edges of $\skel^-(\nu)$. To overcome this issue, we make the following crucial observations. First, we know that $G$ contains at most $2k+2$ switches, otherwise we can safely reject the instance. Consequently, at most $2k+2$ children of  $\nu$ may contain switches different from $u$ and $v$ in their pertinent graphs.  Second, consider now a permutation of the edges of $\skel^-(\nu)$ and the corresponding planar embedding of $G_\nu$. Up to a relabeling of the children, we shall assume that the half-boundary $B_{vu}$ of $\mu_i$ and the half-boundary $B_{uv}$ of $\mu_{i+1}$ form a face $f_i$ of $G_\nu$, for $i=1,\dots,h-1$, and that the external face $f_0$ of $G_\nu$ consists of $B_{uv}$ of $\mu_1$ and $B_{vu}$ of $\mu_h$; see \cref{fig:parallel}. Observe now that each of $u$ and $v$ can contribute at most one angle labeled $+1$ and at most two angles labeled $0$; all other angles at $u$ and $v$ must be labeled  $-1$. Hence, besides the at most six faces in which $u$ or $v$ contribute an angle labeled $+1$ or $0$, all other faces are such that they either contain an angle labeled  $+1$, or all their angles (except those formed by $u$ and $v$) are labeled $0$. Therefore, the number of faces whose half-boundaries have non-empty signatures is at most $t=2(2k+2)+6=4k+10$ (a switch vertex may be bifacial and hence belong to two half-boundaries).  Putting all together, if there exist more than $t$ pairs that are not empty, then we can safely discard the set of candidate tuples. 

    \begin{figure}[t]
        \centering
        \subcaptionbox{\label{fig:parallel}}{\includegraphics{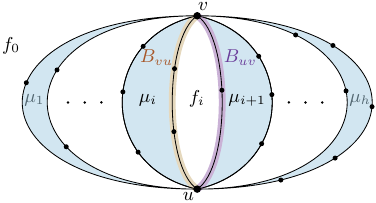}}
        \hfil
        \subcaptionbox{\label{fig:rigid}}{\includegraphics{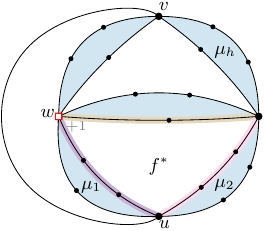}}
        \caption{Illustration for the proof of (a) \cref{le:parallel} and (b) \cref{le:rigid}.}        
    \end{figure}

    Based on the previous observations, we will now assume to have at least $h-t$ empty pairs. Furthermore, if $h>2t+2$, at least two children are such that \cref{le:empty} holds for them. Consequently,  removing all empty pairs except $t+1$ preserves the existence of a solution (if any). Therefore, we shall further assume that we have $h \in O(t) \in O(k)$ pairs of signatures, and we can now enumerate all possible permutations of such pairs, and hence all possible putative planar embeddings described by $\skel^-(\nu)$. 

    Consider now a fixed permutation.  Following the same notation as before, assume that the half-boundary $B_{vu}$ of $\mu_i$ and the half-boundary $B_{uv}$ of $\mu_{i+1}$ form a face $f_i$ of $G_\nu$, for $i=1,\dots,h$. We call such faces \emph{active}. If all values $b_{1,i}$ are true and $u$ is a switch, we guess whether $u$ has an angle labeled $+1$ in some active face $f_i$ or not. In the former case, we set flag $b_1$ to false and also guess which active face the angle belongs to, in the latter case we set $b_1$ to true.  We do the same for $v$ and flag $b_2$. 
     
    Next, consider a non-empty signature containing a symbol $\sigma$ or $\tau$. Let $w$ be the vertex corresponding to that symbol. If $w$ is not bifacial, the active face in which it forms the $+1$ angle is unique, otherwise we must guess in which of the two active faces sharing $w$ the $+1$ angle is assigned to. After doing this procedure for all such symbols, we have exhaustively branched over the $2^{O(k)}$ angle assignments for the active faces. For each such angle assignment we can check, in $O(k)$ time, whether it is an upward assignment for each active face. If not, we discard the angle assignment, otherwise we now have an internal assignment $\lambda$ of $G_\nu$.
    
    Next, for each active face $f_i$, we can  apply \cref{le:saturating} to compute the minimum number  $c_i$ of saturating edges needed to saturate all switches in $f_i$. Let $C+\sum_{i=1}^{h}c_i$ be the cost of the internal assignment $\lambda$. If it is larger than $k$, the angle assignment is discarded. 

    We are now ready to construct the signatures $\Sigma_1$ and $\Sigma_2$ of the half-boundaries $B_{uv}$ and $B_{vu}$ of $\nu$. 
    Since the half-boundary $B_{uv}$ of $\nu$ coincides with $B_{uv}$ of $\mu_1$ (as fixed by the permutation at hand), we can invoke  \cref{le:signature} by using $\lambda$ and $\Sigma_{1,1}$ as arguments. Similarly, the signature $\Sigma_2$ is computed invoking \cref{le:signature} with arguments $\lambda$ and $\Sigma_{2,h}$. 
    Observe that both $\Sigma_1$ and $\Sigma_2$ are short, because $\Sigma_{1,1}$ and $\Sigma_{2,h}$ are short. Then we set $X(\Sigma_1,\Sigma_2,b_1,b_2) = \min \{ X(\Sigma_1,\Sigma_2,b_1,b_2),C+\sum_{i=1}^{h}c_i\}$; taking the minimum is needed because different permutations, as well as different angle assignments of the same permutation, may yield the same pair of signatures and flags but different costs.

    Putting all together, it suffices to first branch over sets of candidate tuples of size $h \in O(k)$, for each set we branch over $k^{O(k)}$ permutations, and for each permutation we further branch over the $2^{O(k)}$ possible angle assignments of the active faces. Computing the cost of an internal assignment takes $2^{O(k^2)} \cdot n$ time by using \cref{le:saturating}.
 \end{proof}

\begin{lemma}\label{le:rigid}
    Let $\nu$ be an R-node of $T$. The set of partial solutions of $\nu$ can be computed in $2^{O(k^2)} \cdot n$ time.
\end{lemma}
\begin{proof}
    Let $u$ and $v$ be the poles of $\nu$, and let $\mu_1,\mu_2, \dots, \mu_h$ be the $h \ge 2$ children of $\nu$. For each child $\mu_i$, let $\langle \Sigma_{1,i}, \Sigma_{2,i}, b_{1,i}, b_{2,i} \rangle$ be a candidate tuple.  Let $C = \sum_{i=1}^h X(\mu_i, \Sigma_{1,i}, \Sigma_{2,i}, b_{1,i}, b_{2,i})$.

    We first verify that $C \le k$, in order to avoid exceeding the budget. Next, we check the consistency of the flags. Recall that the vertices of $\skel(\nu)$ are the poles of the children of $\nu$. Namely, for each vertex $w$ of $\skel(\nu)$, we verify that at most one flag corresponding to it is false. If these conditions are met we proceed as detailed below, otherwise we discard the set of candidate tuples. 

    We now make important observations concerning the number of interesting children of $\nu$. As in the proof of \cref{le:parallel}, we can observe that at most $2k+2$ children of  $\nu$ may contain switches different from $u$ and $v$ in their pertinent graphs. Now consider a child $\mu$ of $\nu$ that does not contain switches in its pertinent graph $G_\mu$, and let $u_\mu$ and $v_\mu$ be its poles. If $G$ admits a solution, one immediately verifies that $G_\mu$ is $st$-planar and its two switches are $u_\mu$ and $v_\mu$. Consequently, in any solution, the two signatures $\Sigma_{u_\mu v_\mu}$ and $\Sigma_{v_\mu u_\mu}$ must be empty. Based on this property, it suffices to consider sets of pairs of signatures in which at most $2k+2$ pairs are not empty.  

    Next, following the lines of the proof of \cref{le:parallel}, consider a non-empty signature containing a symbol $\sigma$ or $\tau$. Let $w$ be the vertex corresponding to that symbol. If $w$ is not bifacial, the face in which it forms the $+1$ angle is unique, otherwise we must guess in which of the two faces sharing $w$ the $+1$ angle is assigned to. This is however not enough for R-nodes. Namely, observe that each face $f^*$ of  $\skel(\nu)^-$ corresponds to a face $f$ of $G_\nu$ whose boundary is formed by one half-boundary for each child of $\nu$ represented by an edge of $f^*$ (which can be a real edge or a virtual edge); see \cref{fig:rigid}. We call such faces \emph{active} in the following. Moreover, the only angles that are not yet defined are those made by the vertices of $\skel(\nu)$ that are switches and whose corresponding flags are all true. For these vertices we shall guess in which active face they make their $+1$ angle.  Clearly, any such a vertex $w$  belongs to multiple active faces (possibly including the external face). On the other hand, for an active face to be able to absorb a $+1$ angle, it must contain at least three angles labeled $-1$. Since we have at most $2k+2$ non-empty pairs, there are at most $4k+4$ active faces formed by non-empty signatures. For the other active faces, the only source of $-1$ angles are the vertices of $\skel(\nu)$. Consequently, if $w$ is incident to more than $4k+5$ active faces in which the number of angles labeled  $-1$ is larger than $2$, we can safely discard the set of candidate tuples. Putting all together, for each vertex $w$ we can branch over its $O(k)$ interesting active faces to decide in which of them it will make its $+1$ angle. This procedure leads to $2^{O(k)}$ angle assignments for the active faces. For each such angle assignment we can check, in $O(k)$ time, whether it is an upward assignment for each of the active faces. If not, we discard the angle assignment, otherwise we now have an internal assignment $\lambda$ of $G_\nu$.

    Next, for each active face $f_i$, we can  apply \cref{le:saturating} to compute the minimum number  $c_i$ of saturating edges needed to saturate all switches in $f_i$. Let $C+\sum_{i=1}^{h}c_i$ be the cost of the internal assignment. If it is larger than $k$, the angle assignment is discarded.

    We are now ready to construct the signatures $\Sigma_1$ and $\Sigma_2$ of the half-boundaries $B_{uv}$ and $B_{vu}$ of $\nu$. 
    Observe that the embedding of $\skel(\nu)$ if fixed up to a flipping operation, which corresponds to inverting the two signatures. Therefore, we construct $\Sigma_1$ and $\Sigma_2$ as follows. Let $\Sigma'_i$, for $i=1,\dots,r$ be the $r \ge 1$ signatures of the half-boundaries of the children of $\nu$ that form the half-boundary $B_{uv}$ of $\nu$, in the order they are encountered from $u$ to $v$. Also let $w_i$, $i=1,\dots,r-1$ be the vertices of $\skel(\nu)$ that belong to $B_{uv}$. We initialize $\Sigma_1$ with the signature obtained by invoking \cref{le:signature} with arguments $\lambda$ and $\Sigma'_1$. For vertex $w_1$, we distinguish whether it is a switch of $G$ or not. In the former case,  we concatenate the symbol $\sigma$ ($\tau$) if none of its angles in $G_\nu$ is labeled as $+1$, otherwise we concatenate $\sigma_\ell$ ($\tau_\ell$). In the latter case, consider the two edges incident to $w_1$ along $B_{uv}$. If one edge is incoming and the other is outgoing, then we do not append any symbol. If both edges are outgoing (incoming), we append $\sigma_\ell$ ($\tau_\ell$). We then repeat the procedure for the remaining signatures and vertices. The signature $\Sigma_2$ is computed analogously. Once both $\Sigma_1$ and $\Sigma_2$ have been computed, we verify that each of them is short (a necessary condition by \cref{le:short}), otherwise we reject the set of candidate tuples. Concerning the flags,  $b_1$ ($b_2$) is true if and only if all flags corresponding to $u$ ($v$) are true and none of its angles in the active faces is labeled $+1$ according to $\lambda$.   Finally we set $X(\Sigma_1,\Sigma_2,b_1,b_2) = \min \{ X(\Sigma_1,\Sigma_2,b_1,b_2),C+\sum_{i=1}^{h}c_i\}$, as well as $X(\Sigma_2,\Sigma_1,b_2,b_1) = \min \{ X(\Sigma_2,\Sigma_1,b_2,b_1),C+\sum_{i=1}^{h}c_i\}$. 

    Putting all together, it suffices to first branch over sets of candidate tuples of size $h \in O(k)$, for each set we branch over the $2^{O(k)}$ possible angle assignments of the active faces. Computing the cost of an internal assignment takes $2^{O(k^2)} \cdot n$ time by using \cref{le:saturating}.
\end{proof}

It remains to deal with the root $\rho$ of $T$. Recall that $G_\rho=G$, and that $\rho$ is a Q-node. 

\begin{lemma}
\label{le:fixedrefedge}
    Let $G$ be an $n$-vertex biconnected digraph, let $e$ be an edge of $G$, and let $k \in \mathrm{N}$. There exists an  algorithm that decides, in $O(2^{O(k^2)} \cdot n)$ time, whether $G$ can be augmented to an  $st$-planar graph with $e$ on its external face by adding at most $k$ edges.
\end{lemma}
\begin{proof}
    By using \cref{le:series,le:parallel,le:rigid} we can traverse $T$ bottom up until reaching the root $\rho$. Let $u,v$ be the end-vertices of $e$, and therefore the poles of both $\rho$ and of its child $\xi$. Consider each pair of signatures $\Sigma'_1$, $\Sigma'_2$ and of flags $b_1,b_2$ such that $C = X(\xi,\Sigma'_1,\Sigma'_2,b_1,b_2) \le k$. Let $f_0$ and $f_1$ be the two faces of $G_\rho=G$ that share edge $e$, with $f_0$ being the external face. Without loss of generality, assume that the boundary of $f_0$ is formed by edge $e$ and by the half-boundary $B_{vu}$ of $\xi$, while the boundary of $f_1$ is formed by $e$ and the half-boundary $B_{uv}$ of $\xi$.  Each symbol $\sigma$ or $\tau$ in $\Sigma_1$ whose corresponding vertex is not bifacial must contribute with an angle labeled  $+1$ in $f_1$. On the other hand, a symbol $\sigma$ or $\tau$ in $\Sigma'_1$ whose corresponding vertex is bifacial must contribute with an angle labeled  $+1$ in one of $f_1$ or $f_0$. For such vertices we guess in which of the two faces they contribute the $+1$ angle. We proceed analogously for $\Sigma'_2$. The same reasoning applies to $u$ ($v$) if $b_1$ ($b_2$) is true. This leads to an $2^{O(k)}$ angle assignments for $f_0$ and $f_1$. 
    For each angle assignment we can test, in $O(k)$ time, whether it is upward for the two faces $f_0$ and $f_1$. 
    
    An angle assignment of $f_0$ and $f_1$ that is upward for them, together with the internal assignment represented by $X(\xi,\Sigma'_1,\Sigma'_2,b_1,b_2)$, implies the existence of an upward angle assignment of $G$. Also, by \cref{le:saturating} we can compute in $2^{O(k^2)}\cdot n$ time the cost of saturating $f_1$. For face $f_0$ we need to saturate all switches except one source and one sink, hence its cost $c_0$ can be computed by a simple adjustment of the procedure of \cref{le:saturating}. 
    If $C^* = C+c_0+c_1 \leq k$, then we have constructed an upward angle assignment of $G$, and, in particular, all switches of $G$ can be saturated with at most $k$ saturating edges, except for a single source and a single sink on the external face. Then we can conclude that $G$ is a YES instance and the algorithm reports a positive answer. 

    After considering all pairs $\Sigma'_1$, $\Sigma'_2$ and all pairs $b_1,b_2$, as well as all angle assignments for the corresponding faces $f_0$ and $f_1$, if no positive answer was returned, then the algorithm halts and rejects the instance.
\end{proof}

The proof of   \cref{th:main-bic} follows by applying  \cref{le:fixedrefedge} for each of the $O(n)$ edges of $G$. 

\section{Discussion and Open Problems}\label{se:open}

We showed that \stshort can be solved in $2 ^{O(k^2)} \cdot n^2$  time for biconnected digraphs. It is worth remarking that, while in principle the \stshort problem needs not to be restricted to biconnected digraphs (for which it is already NP-hard), considering simply connected graphs would make the proof of our result more technical but not more interesting. In fact, one can simply decompose the graph into its biconnected components through a block-cutvertex tree and work with similar boundary conditions as those we already considered. More interestingly, we ask whether \stshort belongs to the FPL (fixed parameter linear) class. On a similar note, improving the exponential function (or proving that it is asymptotically optimal under standard assumptions) would also be interesting. Lastly, it remains open whether \stshort admits a kernel of polynomial size.

\bibliographystyle{plainurl}
\bibliography{bibliography}

\end{document}